\documentclass[10pt]{IEEEtran}
\usepackage{amsmath, mathrsfs,amsfonts}
\usepackage{amssymb}
\usepackage{acronym}
\usepackage{graphicx, amsthm}
\usepackage{pdfsync}
\usepackage{color}

\newcommand{\BEQA}{\begin{eqnarray}}
\newcommand{\EEQA}{\end{eqnarray}}

\def \foral {\textrm{for all }}
\def \pr  {\mathbf{Pr}}
\def \E  {\mathbf{E}}

\newtheorem{theorem}{Theorem}

\newtheorem{corollary}{Corollary}
\newtheorem{lemma}{Lemma}

\theoremstyle{definition}

\newtheorem{definition}{Definition}

               {\begin{list}{}{\leftmargin#1\rightmargin#2\topsep#3}\item{}}%
               {\end{list}}

\title{Avoiding Interruptions - QoE Trade-offs in Block-coded Streaming Media Applications}

\author{Ali ParandehGheibi$^\dag$\thanks{$\dag$ Department of Electrical Engineering and Computer Science, MIT.}, Muriel M\'edard$^\dag$, Srinivas Shakkottai$^\ddag$ \thanks{$\ddag$ ECE Department, Texas A\&M University.}, Asuman Ozdaglar$^\dag$\\
\normalsize{parandeh@mit.edu, medard@mit.edu, sshakkot@tamu.edu, asuman@mit.edu}}

\begin{document}

\maketitle

\begin{abstract}
We take an analytical approach to study Quality of user Experience (QoE) for media streaming applications. We use the fact that random linear network coding applied to blocks of video frames can significantly simplify the packet requests at the network layer and avoid duplicate packet reception. We model the receiver's buffer as a queue with Poisson arrivals and deterministic departures. We consider the probability of interruption in video playback (buffer underflow) as well as the number of initially buffered packets (initial waiting time) as the QoE metrics. We explicitly characterize the optimal trade-off between these metrics by providing upper and lower bounds on the minimum initial buffering required to achieve certain level of interruption probability for different regimes of the system parameters. Our bounds are asymptotically tight as the file size goes to infinity. Further, we show that for arrival rates slightly larger than the play rate, the minimum initial buffering remains bounded as the file size grows. This is not the case when the arrival rate and the play rate match.
\end{abstract}

\section{Introduction}\label{introduction_sec}

Peer-to-peer networks (P2P) are a fast-growing means of video delivery.  It has been estimated that between 35-90\% of Internet bandwidth is consumed by P2P applications~\cite{FraMoo03,GumDun03}.  Today, P2P file-sharing networks are seeing a drop in popularity \cite{Lab09}, but the original file sharing ideas are being used for video streaming in networks such as PPLive \cite{pplive} and QQLive \cite{qqlive}.  As smart phones become the medium of choice for Internet media access, P2P video distribution over the wireless medium is likely to gain significance.

P2P video streaming is generally accomplished by dividing the video file into \emph{blocks,} which are then further divided into packets for transmission.  After each block is received, it can be played out by the receiver.  In order to ensure smooth sequential playback, a fresh block must be received before the current block has been played.  If such a fresh block is not available the playback freezes, causing a negative user experience.  Blocks may be buffered in advance of playing out in order to provide a level of protection against a playback freeze, with more initial buffering providing a lower likelihood of playback interruption. Hence, there is a trade-off between the initial waiting time and playback interruptions.


In this paper, our main objective is to characterize the amount of buffering needed for a target probability of playback interruption over the duration of the video.  We consider a model in which network coding is used across the packets of each block.  A wireless user can obtain coded packets from multiple sources (other users and servers).  However, since the wireless channel is unreliable, packets cannot be obtained deterministically.  Thus, our question is \emph{how much should we buffer prior to playback in order to account for wireless channel variations?}

We first show how to model the receiver's buffer as an M/D/1 queue. We then provide upper and lower bounds on the minimum initial buffering required so that the playback interruption probability is below a desired level. The optimal trade-off between the initial buffering and the interruption probability depends on the file size as well as the arrival rate of the packets as compared to the playback rate. We show that our bounds are asymptotically tight as the file size goes to infinity. Moreover, if the arrival rate is slightly larger than the play rate, the minimum initial buffering for a given interruption probability remains bounded as the file size grows. However, when the arrival rate and the play rate match, the minimum initial buffer size grows as the square-root of the file size.

There is significant work in the space of P2P streaming.   Close to our work, \cite{ZhoChi07,BonMas08,ZhaLuiChi_09,YinSri10} develop analytical models on the trade-off between the steady state probability of missing a block, and buffer size under different block selection policies for live streaming in a full mesh P2P network with deterministic channels.  A further modification is to use random linear network coding techniques \cite{RLNC} to make block selection simpler \cite{Acedanski05, Rodriguez06, wangLi07, ChiZhang06} in the wired and wireless context.  In contrast, we focus on a very different scenario of streaming of pre-prepared content over unreliable wireless channels using network coding.  Further, our analysis is on transient effects---we are interested in the first time that video playback is interrupted as a function of the initial amount of buffering.

\section{System Overview}
We consider a media streaming system as follows.  Media files are usually divided into blocks consisting of multiple frames.  The video coding is such that all the frames in the block need to be available before any frames can be played.  Blocks are requested in sequence by the playback application from the user-end.  The server (or other peers) packetize the requested block and transmit them to the user as in Figure \ref{layers_fig}.  Obtaining the packets of a block from various peers in a P2P system requires the receiver to keep track of missing packets in a block, and request them from different peers.  However, since packet transmission is unreliable in a wireless context, requesting each packet from only one peer might cause unreasonable delays, while requesting a particular packet from multiple peers can result in inefficient resource usage.



\begin{figure}[htbp]
\centering
  \includegraphics[width=.3\textwidth]{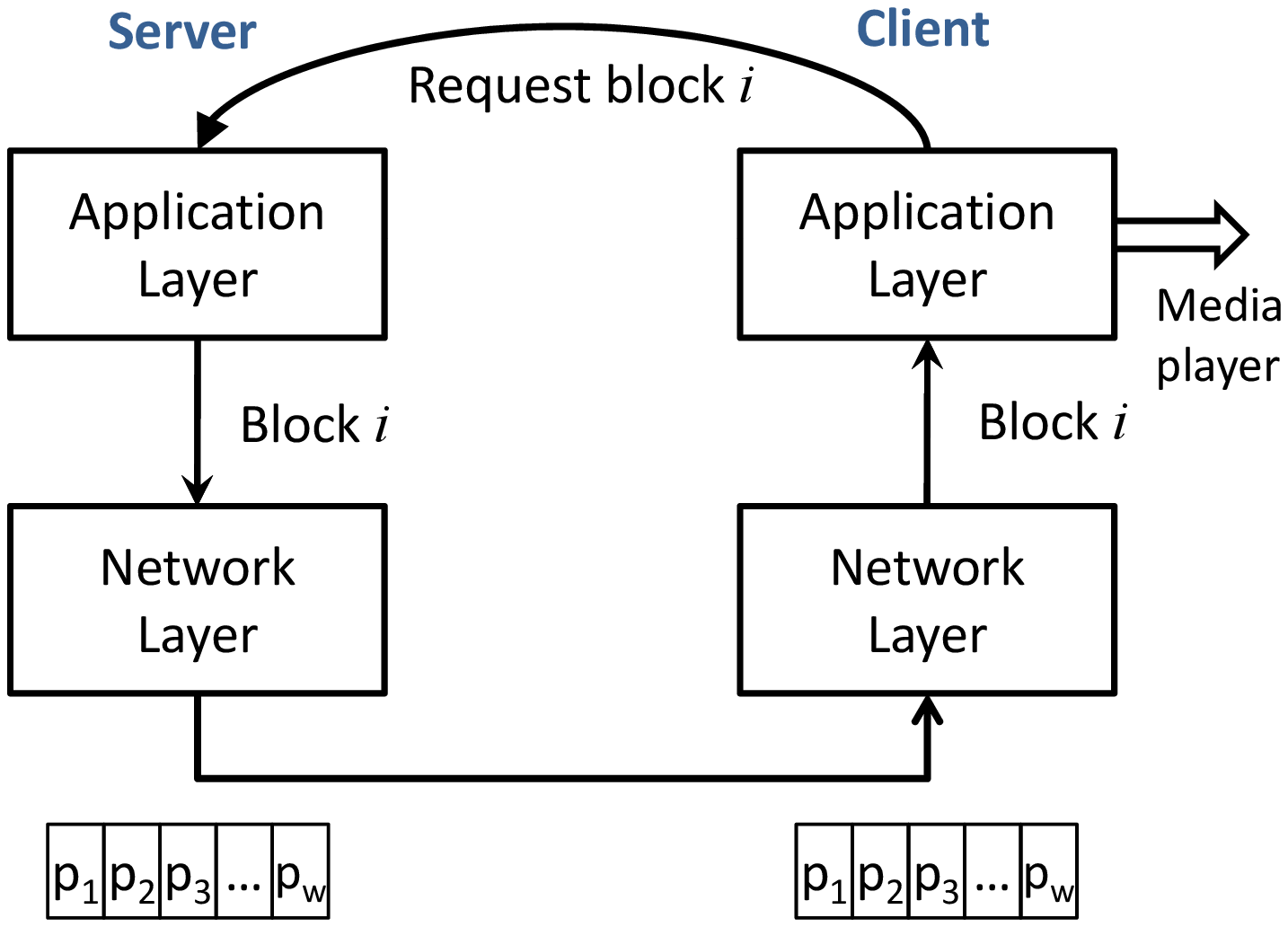}
  \caption{The media player (application layer) requires complete blocks. At the network layer each block is divided into packets and delivered. }\label{layers_fig}
\vspace{-0.1in}
\end{figure}


Random linear codes can be used to alleviate such inefficiencies.  Here, instead of requesting a particular packet from block $i$, the receiver simply requests a \emph{degree of freedom} of block $i.$  The server in turn responds with a random linear combination of all packets that it has in block $i$.  The coefficients of each combination are chosen uniformly at random from a Galois field of size $q$. The coded packets delivered to the receiver can be thought of as linear equations, where the unknowns are the original packets in block $i$. Block $i$ can be fully recovered by solving a system of linear equations if it is full rank. It can be shown that if the field size $q$ is large enough, the received linear equations are linearly independent with very high probability \cite{RLNC}. Therefore, for recovering a block of $W$ packets, it is sufficient to receive $W$ coded packets from different servers.  In a P2P system, it is unlikely that a randomly contacted peer would have all packets corresponding to a particular block.  However, storing blocks in a random linear coded fashion at all peers ensures that with high probability, the selected peer has a new degree of freedom to offer (see  \cite{wangLi07} for further discussion).  Thus, each received coded packet is likely to be independent of previous ones with probability $1-\delta(q),$ where $\delta(q)\rightarrow 0$ as $q \rightarrow \infty$.

Note that such random linear coding does not introduce additional decoding delay for each block, since the frames in a block can only be played out when the whole block is received.   So there is no difference in delay whether the end-user received $W$ uncoded packets of the block or $W$ independent coded packets that can then be decoded.

\section{System Model and QoE Metrics}
Consider a single user receiving a media file from various peers it is connected to. Each peer could be a wireless access point or another wireless user operating as a server. We assume that the video file consists of $T$ packets that are divided into blocks of $W$ packets. Each server sends  random linear combinations of the packets within the current block to the receiver. We assume that the linear combination coefficients are selected from a Galois field of size $q$. We assume the block size $W$ is small compared to the total length of the file, but large enough to ignore the boundary effects of moving from one block to the next.  Time is continuous, and the arrival process of packets from each peer is a Poisson process independent of other arrival processes. Since no redundant packet is delivered from different peers, we can combine the arrival processes into one Poisson process of rate $R_s$. We assume that each received coded packet is linearly independent from the previous ones with probability $1-\delta(q)$. Hence, the effective arrival process of \emph{useful} packets is Poisson with rate $R = R_s(1-\delta(q))$. Note that $R$ approaches $R_s$ for large enough field size. We normalize the playback rate to one, i.e., it takes one unit of time to play a single packet.  Thus, our simplified model is just a single-server-single-receiver system.   We also assume that the parameter $R$ is known at the receiver, which first buffers $D$ packets from the beginning of the file, and then starts the playback.

The presence of some packets in the buffer does not guarantee that there will be no interruption since we require $W$ packets corresponding to a block before it can be decoded and played out.  However, if there are at least $W$ packets in the buffer, there is at least one playable packet.  This is so since either the first $W$ packets in the buffer belong to the same block, or they belong to two different blocks. In the former case, the packets of the block can be decoded, and in the latter case, the first block of the two must be already decoded; otherwise, the next block would not be sent from the server. Therefore, the dynamics of the receiver's buffer size $Q(t)$ can be described as follows
\begin{equation}\label{buffer}
    Q(t) =  D + A(t) - t,
\end{equation}
where $D$ is the initial buffer size and $A(t)$ is a Poisson process of rate $R$.
We declare an interruption in playback when the buffer size decreases to the threshold $W$. For simplicity of notation, we assume that an \emph{extra} block is initially buffered (not taken into account in $D$). Hence, we can declare an interruption in playback when the buffer size reaches zero before reaching the end of the file. More precisely, let
\begin{eqnarray}\label{hitting_def}
\tau_e &=& \inf \{t: Q(t) \leq 0 \}, \nonumber \\
\tau_f &=& \inf \{t: Q(t) \geq T - t \},
\end{eqnarray}
where $\tau_f$ corresponds to time of completing the file download, because we have already played $\tau_f$ packets and the buffer contains the remaining $T-\tau_f$ packets to be played.  The video streaming is interrupted if and only if $\tau_e < \tau_f$.

We consider the following metrics to quantify Quality of user Experience (QoE). The first metric is the initial waiting time before the playback starts. This is directly captured by the initial buffer size $D$. Another metric that affects QoE is the probability of interruption during the playback denoted by
\begin{equation}\label{p_int}
p(D) = \pr \{\tau_e < \tau_f\},
\end{equation}
 where $\tau_e$ and $\tau_f$ are defined in (\ref{hitting_def}). In our model, the user expects to have an interruption-free experience with probability higher than a desired level $1-\epsilon$. Note that there is a  fundamental trade-off  between the interruption probability $\epsilon$ and the initial buffer size $D$. For example, owing to the randomness of the arrival process, in order to have zero probability of interruption, it is necessary to fully download the file, i.e., $D=T$. Nevertheless, we need to buffer only a small fraction of the file if user tolerates a positive probability of interruption. These trade-offs and their relation to system parameters $R$ and $T$ are addressed in the following section.

\section{Optimal QoE Trade-offs}

We would like to obtain the smallest initial buffer size so that the interruption probability is below a desired level $\epsilon$, which is denoted by
\begin{eqnarray}\label{Dmin}
  D^*(\epsilon) &=& \min \{D\geq 0: p(D) \leq \epsilon\},
\end{eqnarray}
where $p(D)$ is the interruption probability defined in (\ref{p_int}). Note that in general $p(D)$ and hence $D^*(\epsilon)$ depend on the arrival rate $R$ and the file size $T$ which are assumed to to be known and constant. In the following we characterize the optimal trade-off between the initial buffer size and the interruption probability by providing bounds on $D^*(\epsilon)$. An upper bound (achievability) on $D^*(\epsilon)$ is particularly useful, since it provides a sufficient condition for desirable user experience. A lower bound (converse) of $D^*(\epsilon)$ provides a necessary condition on the initial buffer size for a desirable level $\epsilon$ of interruption probability. Let us first introduce some useful lemmas.

\begin{lemma}\label{subMG_lemma}
Let $X(t) = e^{-rQ(t)}$, where $Q(t)$ is given by (\ref{buffer}), and define
\begin{equation}\label{gamma_def}
 \gamma(r) = r + R(e^{-r} -1 ).
\end{equation}
Then for every $r\geq 0$ such that $\gamma(r) \geq 0$, $X(t)$ is a sub-martingale with respect to the canonical filtration $\mathcal F_t = \sigma(X(s), 0\leq s \leq t)$, i.e., the smallest $\sigma$-field containing the history of the stochastic process $X$ up to time $t$.
Moreover, if $\gamma(r) =0$ then $X(t)$ is a martingale.
\end{lemma}
\begin{proof}
For every $t$, $|X(t)| \leq 1$. Hence, $X(t)$ is uniformly integrable. It remains to show that for every $t\geq 0$ and $h>0$,
\begin{equation}\label{subMG_def}
\E[X(t+h)|\mathcal F_t] \geq X(t) \quad \textrm{a.s.}
\end{equation}
$X(t)$ is a martingale if (\ref{subMG_def}) holds with equality.
The left-hand side of (\ref{subMG_def}) can be expressed as
\begin{eqnarray}
  \E[X(t+h)|\mathcal F_t] &=& \E\Big[e^{-r(Q(t+h) - Q(t))}\Big|\mathcal F_t\Big] X(t) \nonumber \\
    &=& \E\Big[e^{-r(A(t+h) - A(t))}\Big|\mathcal F_t\Big] e^{rh} X(t)  \nonumber \\
    &\stackrel{(a)}{=}& \E\big[e^{-rA(h)}\big] e^{rh} X(t)  \nonumber \\
    &\stackrel{(b)}{=}& e^{h(r+R(e^{-r}-1))} X(t) = e^{h\gamma(r)} X(t),\nonumber
\end{eqnarray}
 where  (a) follows from independent increment property of the Poisson process, and (b) follows from the fact that $A(t)$ is a Poisson random variable. Now,  it is immediate to verify (\ref{subMG_def}) for any $r$ with $\gamma(r) \geq 0$. Finally, note that if $\gamma(r) =0$, the equality in the above relations hold through, and (\ref{subMG_def}) holds with equality. Therefore, $X(t)$ is a martingale for $r$ with $\gamma(r) = 0$.
\end{proof}

Next, we use Doob's maximal inequality \cite{SP_book} to bound the interruption probability.
\begin{lemma}\label{pub_lemma}
Let $p(D)$ be the interruption probability defined in (\ref{p_int}), and $\gamma(r)$ be given by (\ref{gamma_def}). Then, for any $r\geq 0$ with $\gamma(r) \geq 0$
\begin{equation}\label{p_ub_gen}
    p(D) \leq  e^{-rD + T\gamma(r)}, \quad \foral D, T, R \geq 0.
\end{equation}
\end{lemma}
\begin{proof}
By definition of $p(D)$ in (\ref{p_int}), we have
\begin{eqnarray*}
  p(D) &=& \pr\{\tau_e < \tau_f\} \\
   &\leq& \pr\{\tau_e \leq T\} = \pr\Big\{\inf_{0\leq t\leq T}Q(t) \leq 0\Big\}  \\
   &=&  \pr\Big\{\sup_{0\leq t\leq T}e^{-rQ(t)} \geq 1\Big\} \\
   &\stackrel{(a)}{\leq}& \E[e^{-rQ(T)}] = \E[e^{-r(D+A(T)-T)}] \\
   &=& e^{-r(D-T)} e^{RT(e^{-r}-1)} = e^{-rD+T\gamma(r)},
\end{eqnarray*}
where (a) holds by applying Doob's maximal inequality \cite{SP_book} to the non-negative sub-martingale $X(t) = e^{-rQ(t)}$. Note that $X(t)$ is a sub-martingale for all $r$ with $\gamma(r)\geq 0$ by Lemma \ref{subMG_lemma}.
\end{proof}

\begin{lemma}\label{r_bar_lemma}
Let $\gamma(r)$ be as defined in (\ref{gamma_def}). Define $\bar r(R)$ as the largest root of $\gamma(r)$, i.e.,
\begin{equation}\label{r_bar}
    \bar r(R) = \sup\{r: \gamma(r) = 0 \}.
\end{equation}
The following relations hold:
\begin{eqnarray}
  \bar r(R) = 0, \quad &&\textrm{if \ } 0\leq R \leq 1, \ \ \label{case1}\\
  \frac{2(R-1)}{R} \leq \bar r(R) \leq 2(R-1), \quad &&\textrm{if \ } 1\leq R \leq 2, \ \ \label{case2} \\
  R-1 \leq \bar r(R) \leq R \leq 2(R-1) , \quad &&\textrm{if \ } R \geq 2.\ \  \label{case3}
\end{eqnarray}
\end{lemma}
\begin{proof}
We omit the proof for brevity. See \cite{ISIT_report}.
\end{proof}

Next, we provide sufficient conditions on the initial buffer size to avoid interruptions with high probability for different regimes of the arrival rate.

\begin{theorem}\label{Dub_thm}
\emph{[Achievability]} Let $D^*(\epsilon)$ be defined as in (\ref{Dmin}), and $\bar r(R)$ be given by (\ref{r_bar}). Then
\begin{description}
  \item[(a)] For all $R>1$,
  \begin{equation}\label{Dub_a}
    D^*(\epsilon) \leq \frac{1}{\bar r(R)}\log\big(\frac1\epsilon\big).
  \end{equation}
  \item[(b)] For all $0\leq R \leq 1+ \Big(\frac{1}{2T}\log\big(\frac1\epsilon\big)\Big)^{\frac12}$,
  \begin{eqnarray}\label{Dub_b}
    D^*(\epsilon) &\leq& \min\Big\{\frac{1}{\bar r(R)}\log\big(\frac1\epsilon\big), \nonumber \\  && T(1-R)+\Big({2TR\log\big(\frac1\epsilon\big)}\Big)^{\frac12} \Big\}. \qquad
  \end{eqnarray}
\end{description}
\end{theorem}

\begin{proof}
First, note that for any upper bound $\bar p(D)$ of the interruption probability $p(D)$, any feasible solution of
\begin{equation}\label{Dmin_ub}
\bar D(\epsilon) =  \min \{D\geq 0: \bar p(D) \leq \epsilon\}
\end{equation}
provides an upper bound on $D^*(\epsilon)$. This is so since the optimal solution of the above problem is feasible in the minimization problem (\ref{Dmin}). If the problem in (\ref{Dmin_ub}) is infeasible, we use the convention $\bar D(\epsilon) = \infty$, which is a trivial bound on $D^*(\epsilon)$. The rest of the proof involves finding the tightest bounds on $p(D)$ and solving (\ref{Dmin_ub}).

\emph{Part (a):} By Lemma \ref{pub_lemma}, for $r = \bar r(R)$, we can write
$$p(D) \leq \bar p_a(D) =e^{-\bar r(R) D}, \quad \foral D, T, R \geq 0.$$

Solving $\bar p_a(D) = \epsilon$ for $D$ gives the result of part (a). Since $\bar r(R) = 0$ for $R\leq 1$ (cf. Lemma \ref{r_bar_lemma}), this bound is not useful in that range.

\emph{Part (b):} First, we claim that for all $D \geq T(1-R+\bar r(R))$,
$$p(D) \leq \bar p_b(D) = e^{-\frac12 TR z^2},$$
where $z = 1 - \frac1R\big(1-\frac{D}{T}\big)$. We use Lemma \ref{pub_lemma} with $r= r^* = -\log\big(\frac1R\big(1-\frac{D}{T}\big)\big)$ to prove the claim. Note that $r^* \geq 0$, because $D\geq T(1-R)$. In order to verify the second hypothesis of Lemma \ref{pub_lemma}, consider the following
\begin{eqnarray*}
  R(e^{-r^*} - e^{-\bar r(R)}) &=&   \bar r(R) + R(e^{-r^*}-1) - \gamma( \bar r(R)) \\
   &\stackrel{(a)}{=}&   \bar r(R) -R + (1-\frac{D}{T}) \\
   &\stackrel{(b)}{=}&   \frac1T \Big[ T(1-R +\bar r(R))-D \Big] \stackrel{(c)}{\leq} 0,
   \end{eqnarray*}
where (a) and (b) follow from the definition of $\bar r(R)$ and $r^*$, respectively, and (c) holds by the hypothesis of the claim. Thus, $r^* \geq \bar r(R)$. Using the facts that  $\bar r(R)$ is the largest root of $\gamma(r)$, and $\gamma(r) \rightarrow +\infty$ as $r \rightarrow \infty$, we conclude that $\gamma(r^*) \geq 0$. Now, we apply Lemma \ref{pub_lemma} to get
\begin{eqnarray*}
  p(D) &\leq& e^{-r^*D + T\gamma(r^*)} \\
   &\stackrel{(a)}{=}& e^{TR\big(\frac1R(1-\frac{D}{T})r^* - (1- e^{-r^*})\big)}  \\
   &\stackrel{(b)}{=}& e^{TR\big(-(1-z)\log(1-z)-z\big)}  \\
   &\stackrel{(c)}{\leq}& e^{-\frac12 TRz^2},
\end{eqnarray*}
where (a) and (b) follow from the definition of $\gamma(r)$ and $z$. We skip the proof of (c) for brevity (cf. Appendix of \cite{ISIT_report}). Therefore, the claim holds.

Now, let $\bar D = T(1-R)+\Big({2TR\log\big(\frac1\epsilon\big)}\Big)^{\frac12}.$ Using the claim that we just proved, we may verify that $p(\bar D) \leq \bar p_b(\bar D) = \epsilon$, if $\bar D \geq T(1-R+\bar r(R))$. In order to check the hypothesis of the claim, note that for $R \leq 1$, $\bar r(R) = 0$ (cf.  Lemma \ref{r_bar_lemma}), and for all $1 \leq R \leq 1+ \Big(\frac{1}{2T}\log\big(\frac1\epsilon\big)\Big)^{\frac12}$, we have
\begin{eqnarray*}
  \bar D - T(1-R) &=&  \Big({2TR\log\big(\frac1\epsilon\big)}\Big)^{\frac12} \\
  &\geq& 2T\Big({\frac{1}{2T}\log\big(\frac1\epsilon\big)}\Big)^{\frac12}\\
  &\stackrel{(d)}{\geq}& 2T(R-1) \stackrel{(e)}{\geq} T\bar r(R),
\end{eqnarray*}
where inequality (d) follows from the hypothesis of Part (b), and inequality (e)  is true by Lemma \ref{r_bar_lemma}. Therefore, $D^*(\epsilon) \leq \bar D$ for all $R \leq 1+ \Big(\frac{1}{2T}\log\big(\frac1\epsilon\big)\Big)^{\frac12}$. Note that, the upper bound that we obtained in Part (a) is also valid for all $R$. Hence, the minimum of the two gives the tightest bound.
\end{proof}

When the arrival rate $R$ is smaller than one (the playback rate), the upper bound in Theorem \ref{Dub_thm} consists of two components. The first term, $T(1-R)$, compensates the expected number of packets that are required by the end of $[0,T]$ period. The second component, $\Big({2TR\log\big(\frac1\epsilon\big)}\Big)^{\frac12} $, compensates the randomness of the arrivals to avoid interruptions with high probability. Note that this term increases by decreasing the maximum allowed interruption probability, and it would be zero for a deterministic arrival process. For the case when the arrival rate is larger than the playback rate, the minimum required buffer size does not grow with the file size. By continuity of the probability measure, we can show that the  upper bound  in Theorem \ref{Dub_thm} remains bounded for infinite file sizes. This is so since the buffer size in (\ref{buffer}) has a positive drift. Hence, if there is no interruption at the beginning of the playback period, it becomes more unlikely to happen later.

In the following, we show that the upper bounds presented in Theorem \ref{Dub_thm} are \emph{asymptotically tight}, by providing lower bounds on the minimum required buffer size $D^*(\epsilon)$, for different regimes of the arrival rate $R$. Let us first define the notion of a tight bound.

\begin{definition}\label{tight_def}
Let $\hat D$ be a lower or upper bound of the minimum buffer size $D^*(\epsilon)$ that depends on the file size $T$.  The bound $\hat D$ is an \emph{asymptotically tight}  bound if $\frac{|\hat D - D^*(\epsilon)|}{D^*(\epsilon)}$ vanishes as $T$ goes to infinity.
\end{definition}

\begin{figure}
\centering
  \includegraphics[width=3in]{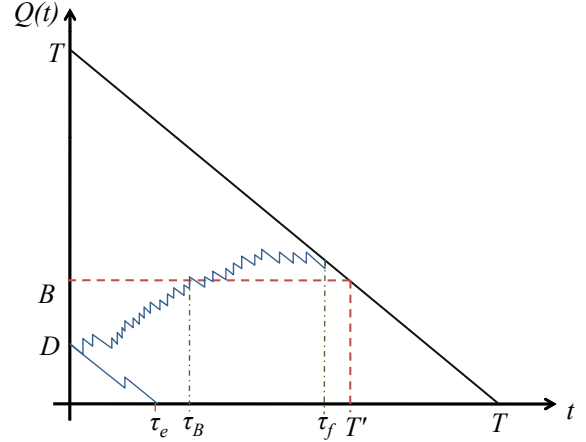}
  \caption{Two sample paths of the buffer size $Q(t)$ demonstrating the interruption event at time $\tau_e$, crossing the threshold $B$ at time $\tau_B$, and the download complete event at time $\tau_f$.}\label{Q_t_fig}
\vspace{-0.2in}
\end{figure}

\begin{theorem}\label{Dlb_thm}
\emph{[Converse]} Let $D^*(\epsilon)$ be defined as in (\ref{Dmin}), and $\bar r(R)$ be given by (\ref{r_bar}). Then
\begin{description}
  \item[(a)]  For all $R>1$,
  \begin{equation}\label{Dlb_a}
    D^*(\epsilon) \geq -\frac{1}{\bar r(R)}\log\Big( \epsilon + 2e^{-\frac{(R-1)^2}{4(R+1)}T}  \Big).
  \end{equation}
  \item[(b)]  For each $R \leq 1$ and $\epsilon \leq \frac{1}{16}$, if $T \geq C \log\big(\frac{1}{\epsilon}\big)$ then
  \begin{eqnarray}\label{Dlb_b}
    D^*(\epsilon) &\geq& T(1-R)+\frac12 \Big({2TR\log\big(\frac1\epsilon\big)}\Big)^{\frac12},\qquad
  \end{eqnarray}
  where $C$ is a constant that only depends on $R$.
\end{description}
\end{theorem}

\begin{proof}
\emph{Part(a):} Similarly to the argument as in the proof of Theorem \ref{Dub_thm}, it is sufficient to provide a lower bound on $p(D)$ defined in (\ref{p_int}).

Define $\tau_B$ as the first time that $Q(t)$ crosses a threshold $B>D$, i.e.,
$$\tau_B = \inf \{t: Q(t) \geq B \}.$$

A necessary condition for the interruption event to happen is to have the receiver's buffer emptied before time $T' = T-B$, or crossing the threshold $B$ (see Figure \ref{Q_t_fig}). In particular,
\begin{equation}\label{p_lb}
p(D) = \pr \{\tau_e < \tau_f\} \geq \pr \big\{\tau_e \leq \min\{\tau_B, T'\}\big\}.
\end{equation}

Define the stopping time
$$\tau = \min\{\tau_e, \tau_B, T'\},$$
and let $Y(t) = e^{-\bar r(R) Q(t)}$, where $\bar r(R) >0$ is given by (\ref{r_bar}). By Lemma \ref{subMG_lemma}, $Y(t)$ is a martingale. Moreover, $Y(t) \leq 1$, and $\tau \leq T < \infty$. Therefore, we can apply Doob's optional stopping theorem \cite{SP_book} to get
\begin{eqnarray}
  e^{-\bar r(R) \cdot D} &=& \E[Y(0)] \stackrel{(a)}{=} \E[Y(\tau)] \nonumber \\
  &\stackrel{(b)}{\leq}& e^{-\bar r(R) \cdot 0} \big(\pr\{\tau = \tau_e\} + \pr \{\tau = T'\}\big) \nonumber \\
  && + e^{-\bar r(R) \cdot B} \big(1- \pr\{\tau = \tau_e\} - \pr \{\tau = T'\}\big) \nonumber \\
  &\leq& \pr\{\tau = \tau_e\}  + \pr \{\tau = T'\}+ e^{-\bar r(R) \cdot B} \nonumber \\
  &\stackrel{(c)}{\leq}& p(D) +  \pr \{\tau = T'\} + e^{-\bar r(R) \cdot B}, \nonumber
\end{eqnarray}
where (a) is the result of Doob's optional stopping time. (b) holds because $Y(t) \leq 1$ for all $t$, and $Y(t) \leq e^{-\bar r(R) \cdot B}$ if $Q(t) \geq B$. Finally, (c) follows from (\ref{p_lb}). Rearranging the terms in the above relation, we obtain
\begin{equation}\label{p_lb2}
p(D) \geq e^{-\bar r(R)D} - e^{-\bar r(R)B} -  \pr \{\tau = T'\}.
\end{equation}

Now, choose $B = (1-\alpha)T$, where $\alpha = \frac{R+1}{2R} > \frac{1}{R}$ for all $R>1$.  For all $D,T \geq 2$, we have
\begin{eqnarray}\label{QT'_ub}
 \pr \{\tau = T'\} &\stackrel{(a)}{\leq}& \pr\{0 \leq Q(T-B) \leq B\} \nonumber \\
  &\stackrel{(b)}{=}& \pr\{\alpha T - D \leq A(\alpha T) \leq T - D\} \nonumber  \\
  &\leq& \pr\big\{ A(\alpha T) \leq R\alpha T - ((R\alpha -1 )T + D)\big\} \nonumber \\
  &\stackrel{(c)}{\leq}& \exp\bigg(-\frac{1}{2R\alpha T} \big((R\alpha -1)T + D - \frac32 \big)^2\bigg) \nonumber \\
  &\stackrel{(d)}{\leq}& \exp\bigg(-\frac{(R-1)^2}{4(R+1)}T\bigg),
\end{eqnarray}
where (a) holds because $Q(t)$ cannot be negative or above the threshold $B$ if stopping at $T'$, and (b) follows from the buffer dynamics in (\ref{buffer}). Recall that $A(\alpha T)$ is a Poisson random variable with mean $R\alpha T$. Since $\alpha \geq \frac1R$, (c) holds for $D, T\geq 2$ by employing Lemma \ref{Poisson_tail} with
$$\lambda = R\alpha T \geq T \geq 2, \quad k = (R\alpha -1 )T + D \geq D \geq 2.$$
Finally, (d) is immediate by definition of $\alpha$ noting that $D \geq 2$.

By Lemma \ref{r_bar_lemma}, we have $\bar r(R) \geq (R-1)$ for all $R > 1$. Therefore, we can bound the second term in (\ref{p_lb2}) as follows
\begin{eqnarray}\label{erB_ub}
  e^{-\bar r(R)B} &\leq& \exp\bigg(- (R-1)\frac{R-1}{2R} T \bigg) \nonumber \\
  &\leq& \exp\bigg(-\frac{(R-1)^2}{4(R+1)}T\bigg).
\end{eqnarray}

Combine the bounds in (\ref{QT'_ub}) and (\ref{erB_ub}) with (\ref{p_lb2}) to obtain
\begin{equation}\label{p_lb3}
    p(D) \geq e^{-\bar r(R)D} - 2e^{-\frac{(R-1)^2}{4(R+1)}T}, \quad \foral D,T \geq 2.
\end{equation}

Therefore, $p(D) \geq \epsilon$ if $ D = -\frac{1}{\bar r(R)}\log\Big( \epsilon + 2e^{-\frac{(R-1)^2}{4(R+1)}T} \Big) \geq 2$. This immediately gives the result in (\ref{Dlb_a}). For the case in which $D<2$ or $T<2$, the claim holds trivially.

\begin{figure}
\centering
  \includegraphics[width=3in]{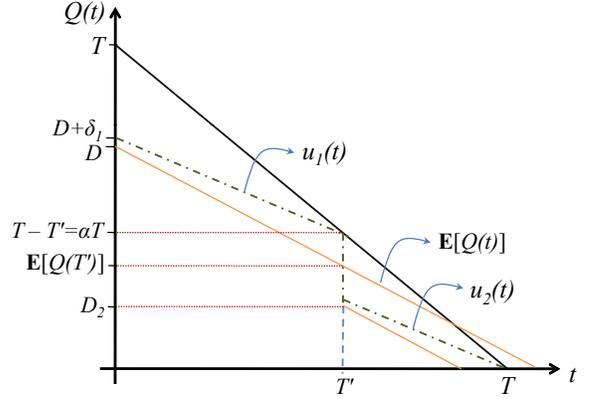}
  \caption{Guideline for the proof of Theorem \ref{Dlb_thm}(b).}\label{Q_LB_fig}
\vspace{-0.2in}
\end{figure}

\emph{Part(b):} It is sufficient to show $p(D) \geq \epsilon$ for $\epsilon \leq \frac{1}{16}$ and $T$ large enough where
\begin{equation}\label{Dlb_def}
D =  T(1-R)+\frac12 \Big({2TR\log\big(\frac1\epsilon\big)}\Big)^{\frac12}.
\end{equation}

Let us first define the boundary functions
\begin{eqnarray}
  u_1(t) &=& D + \delta_1 + (R+\frac{\delta_1}{T'} -1)t, \ t \in [0, T'], \nonumber \\
  u_2(t) &=& D_2 + \delta_2 + (R+\frac{\delta_2}{T-T'} -1)(t-T'), \ t \in [T', T], \nonumber 
\end{eqnarray}
where $T' = (1-\alpha)T$ for some constant $0 < \alpha <1$, and
\begin{eqnarray}
  \delta_1 &=& \frac12 \alpha TR - \frac14 \Big({2TR\log\big(\frac1\epsilon\big)}\Big)^{\frac12},  \label{delta1}\\
  \delta_2 &=& \Big({\alpha TR\log\big(\frac1\epsilon\big)}\Big)^{\frac12}, \label{delta2} \\
  D_2 &=& \alpha T(1-R) - 2\Big({\alpha TR\log\big(\frac1\epsilon\big)}\Big)^{\frac12}. \label{D2}
\end{eqnarray}
Also, denote by $\tau_i$ the first time to hit the boundary $u_i(t)$, i.e.,
$$\tau_i = \inf \{t: Q(t) \geq u_i(t) \}, \quad i =1,2.$$

Observe that for every sample path of the buffer size $Q(t)$, the only way for completing the file download, $Q(t) = T-t$, is to cross the boundary functions $u_i(t)$ or have $Q(T') \geq D_2$ (cf. Figure \ref{Q_LB_fig}). This gives a necessary condition for the interruption event. Hence, 
\begin{eqnarray}
  p(D) &=& \pr \{\tau_e < \tau_f\}\nonumber \\
   &\geq& \pr \{\tau_e \leq \min\{\tau_1, T'\}\}\nonumber \\
  &&  + \pr\{\tau_e \leq \tau_2, T' \leq \min\{\tau_e, \tau_1\} , Q(T') \leq D_2 \} \nonumber \\
  &=& \pr \{\tau_e \leq \min\{\tau_1, T'\}\}\nonumber \\
   && +\pr\{\tau_e \leq \tau_2 |  T' \leq \min\{\tau_e, \tau_1\} , Q(T') \leq D_2\} \big[1 \nonumber \\
  && - \pr\{\tau_1 \leq \min\{\tau_e, T'\} \}  - \pr\{\tau_e \leq \min\{\tau_1, T'\} \}\nonumber \\
  &&- \pr\{T' \leq \min\{\tau_e, \tau_1\} , Q(T') > D_2\}\big] \nonumber \\
  &\geq& \pr\{\tau_e \leq \tau_2 |  T' \leq \min\{\tau_e, \tau_1\} , Q(T') \leq D_2\} \big[1 \nonumber \\
  && - \pr\{\tau_1 \leq \min\{\tau_e, T'\} \} \nonumber \\
  &&- \pr\{T' \leq \min\{\tau_e, \tau_1\} , Q(T') > D_2\}\big]. \label{p_lb_terms}
\end{eqnarray}

In the following we provide bounds on each of the terms in (\ref{p_lb_terms}). By Markov property of the Poisson process we have
\begin{eqnarray}
&&  \pr\{\tau_e \leq \tau_2 |  T' \leq \min\{\tau_e, \tau_1\} , Q(T') \leq D_2\}  \nonumber \\
&&  = \pr\{\tau_e \leq \tau_2 | Q(T') \leq D_2\}  \nonumber \\
&& \geq \pr\{\tau_e \leq \tau_2 | Q(T') = D_2\} \nonumber \\
&& = 1- \pr\{\tau_2 \leq \tau_e | Q(T') = D_2\} \nonumber \\
&& \geq 1- \pr\Big\{\sup_{T'\leq t\leq T}Q(t) \geq u_2(t)  | Q(T') = D_2\Big\} \nonumber \\
&& \geq 1 - \exp\bigg(-\frac{\delta_2^2}{R(T-T')}\bigg) = 1-\epsilon, \label{plb_term1}
\end{eqnarray}
where the last inequality follows from Lemma \ref{div_lemma} with parameters
$\delta = \delta_2$ and $s = \frac{\delta_2}{T-T'}$, if its hypothesis $s\leq R$ is satisfied. This is equivalent to having $T$ satisfy
\begin{equation}\label{T1}
    T \geq \frac{1}{\alpha R}\log\big(\frac1\epsilon\big).
\end{equation}

Similarly, by employing Lemma \ref{div_lemma} with $\delta = \delta_1$ and $s = \frac{\delta_1}{T'}$ we have
\begin{eqnarray}
  && \pr\{\tau_1 \leq \min\{\tau_e, T'\} \} \leq  \pr\Big\{\sup_{0\leq t\leq T'}Q(t) \geq u_1(t) \Big\} \nonumber  \\
  &&\leq  \exp\bigg(-\frac{\delta_1^2}{RT'}\bigg) \nonumber \\
  && = \exp\bigg(-\frac{\alpha^2RT}{4(1-\alpha)}\Big[1-\Big(\frac{\log\big(\frac1\epsilon\big)}{2\alpha^2RT}\Big)^{\frac12}\Big]^2\bigg). \nonumber \end{eqnarray}
Note that the hypothesis of Lemma \ref{div_lemma} is satisfied here for all $\alpha \leq \frac23$. Moreover, if
\begin{equation}\label{T2}
    T \geq \frac{16}{\alpha^2 R}\log\big(\frac1\epsilon\big),
\end{equation}
we have
\begin{eqnarray}
 \pr\{\tau_1 \leq \min\{\tau_e, T'\} \}  &\leq& \exp\bigg(-4\log\big(\frac1\epsilon\big)[1-\frac{1}{\sqrt{32}}]^2\bigg) \nonumber\\
 &\leq& \exp\Big(-2\log\big(\frac1\epsilon\big)\Big) = \epsilon^2.  \label{plb_term2}
\end{eqnarray}

For the last term in (\ref{p_lb_terms}) write
\begin{eqnarray}
&& 1 - \pr\{T' \leq \min\{\tau_e, \tau_1\} , Q(T') > D_2\} \nonumber \\
&& \geq 1 - \pr\{Q(T')> D_2\} = \pr\{Q(T') \leq D_2\} \nonumber \\
&& = \pr\{D + A(T') - T' \leq D_2\} \nonumber \\
&& = \pr\{A(T') \leq T'R -m\sqrt{T'R} \}, \label{plb_term3}
\end{eqnarray}
where $A(T')$ is a Poisson random variable with mean $RT'$, and 
\begin{equation}\label{m_def}
m = \Big(\frac{\log\big(\frac1\epsilon\big)}{2(1-\alpha)}\Big)^{\frac12} \big[1 +(8\alpha)^{\frac12}\big].
\end{equation}

 If $T$ satisfies (\ref{T2}) and  $\alpha, \epsilon \leq \frac{1}{16}$, we may verify that
$$m \leq \sqrt{RT'}/20 - 1.$$
Hence, we can use Lemma \ref{Poisson_tail_lb} to bound (\ref{plb_term3}) from below and conclude
$$1 - \pr\{T' \leq \min\{\tau_e, \tau_1\} , Q(T') > D_2\} \geq \frac13 e^{-\frac{1}{1.9}{\big(m +\frac12\big)^2}}.$$
Observe that for $\alpha = 0$, we have $m = m_0 = \Big(\frac12 \log\big(\frac1\epsilon\big)\Big)^{\frac12}$ and verify that 
$$\frac{1}{1.9}{\big(m_0 +\frac12\big)^2} < \log\big(\frac1\epsilon\big) - \log\big(\frac{17}{15}\big) ,\quad \foral \epsilon \leq \frac{1}{16}.$$
By continuity of $m$ in $\alpha$ (cf. (\ref{m_def})), we can choose the parameter $\alpha = \alpha_0 > 0$ small enough such that
\begin{equation}\label{plb_term4}
    \frac13 e^{-\frac{1}{1.9}{\big(m +\frac12\big)^2}} \geq \frac{17}{15}\epsilon.
\end{equation}

Now, by plugging this relation as well as the preceding bounds in (\ref{plb_term1}) and (\ref{plb_term2}) back in (\ref{p_lb_terms}) we have for all $\epsilon \leq \frac{1}{16}$,
 \begin{equation}\label{plb_sum}
    p(D) \geq (1-\epsilon)\big(\frac{17}{15}\epsilon - \epsilon^2\big) \geq \frac{15}{16} \big(\frac{17}{15} - \frac{1}{16}\big)\epsilon \geq \epsilon,
 \end{equation} 
 if $T\geq  \frac{16}{\alpha_0^2 R}\log\big(\frac1\epsilon\big)$. Therefore, the initial buffer size $D$, defined in (\ref{Dlb_def}), is a lower bound on $D^*(\epsilon)$.
\end{proof}

Note that the result in part (b) of Theorem \ref{Dlb_thm} does not hold for all $\epsilon$. In fact, we can show that $D^*(\epsilon) < T(1-R)$ for a large interruption probability $\epsilon$. In the extreme case $\epsilon =1$, it is clear that $D^*(\epsilon) =0$. Nevertheless, since we are interested in \emph{avoiding} interruptions, we do not study this regime of the interruption probabilities. Comparing the lower bounds obtained in Theorem \ref{Dlb_thm} with the upper bounds obtained in Theorem \ref{Dub_thm}, we observe that they demonstrate a similar behavior as a function of the parameters $T$ and $R$. Now, we can show that the obtained bounds are asymptotically tight.

\begin{corollary}
The upper bounds and lower bounds of $D^*(\epsilon)$ given by Theorems \ref{Dub_thm} and \ref{Dlb_thm} are asymptotically tight, if $R > 1$, or $R < 1$ and $\epsilon \leq \frac{1}{16}$.
\end{corollary}
\begin{proof}
Let $D_l$ and $D_u$ be lower and upper bounds of $D^*(\epsilon)$, respectively. By Definition \ref{tight_def}, for $D_l$ or $D_u$ to be asymptotically tight, it is sufficient to show $\frac{D_u - D_l}{D_l}$ goes to zero as $T$ grows. We may verify this claim by using the upper and lower bounds presented in Theorem \ref{Dub_thm} and Theorem \ref{Dlb_thm}, and taking the limit as $T$ goes to infinity.
%
%
\end{proof}
Next, we numerically obtain the optimal trade-off curve between the interruption probability and initial buffer size, and compare the results with the bounds derived earlier.

\section{Numerical Results}

We use MATLAB simulations to compute the minimum initial buffer size $D^*(\epsilon)$ for a given interruption probability $\epsilon$ in various scenarios. We start from a small initial buffer size $D$, and for each $D$ we compute the interruption probability $p(D)$ via Monte-Carlo method. We increase $D$ until the constraint $p(D) \leq \epsilon$ is satisfied. Since $p(D)$ is monotonically decreasing in $D$, this gives the minimum required buffer size. Here, we restrict $D$ to take only integer values, and round each upper bound value up to the nearest integer, and each lower bound value down to the nearest integer.

\begin{figure}[htbp]
\centering
\vspace{-0.15in}
  \includegraphics[width=3in]{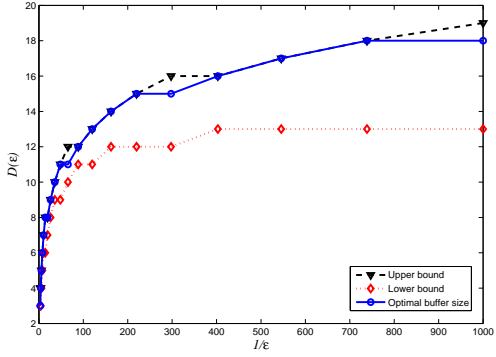}
  \caption{The minimum  buffer size $D^*(\epsilon)$ as a function of the interruption probability. }\label{D_e_fig}
\vspace{-0.1in}
\end{figure}

Figure \ref{D_e_fig} shows the minimum required buffer size $D^*(\epsilon)$  as well as the upper and lower bounds given by Theorems \ref{Dub_thm} and \ref{Dlb_thm} as a function of $\frac1\epsilon$, where the arrival rate is fixed to $R = 1.2$ and the file size $T = 500$. We observe that the numerically computed trade-off curve closely matches our analytical results.

\begin{figure}[htbp]
\centering
\vspace{-0.15in}
  \includegraphics[width=3in]{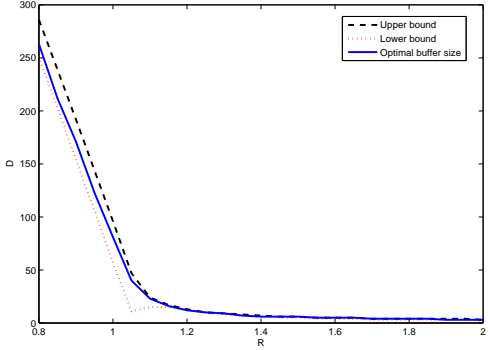}
  \caption{The minimum  buffer size $D^*(\epsilon)$ as a function of the arrival rate $R$. }\label{D_R_fig}
\vspace{-0.1in}
\end{figure}

Figure \ref{D_R_fig} plots the minimum required buffer size $D^*(\epsilon)$  as well as the upper and lower bounds given by Theorems \ref{Dub_thm} and \ref{Dlb_thm} versus the arrival rate $R$, where $\epsilon = 10^{-2}$ and the file size is fixed to $T = 10^3$. Note that when the arrival rate is almost equal or less than the playback rate, increasing the arrival rate can significantly reduce the initial buffering delay. However, for larger arrival rates $D^*(\epsilon)$ is small enough such that increasing $R$ does not help anymore. 

\section{Conclusions}\label{conclusion_sec}

In this paper, we studied the problem of media streaming with focus on the trade-offs between the two QoE metrics---probability of interruption in media playback, and the initial waiting time before starting the playback.  In our system, the user can receive packets of the media stream from multiple sources by requesting packets in each block of the file. We used the fact that sending random linear combinations of the packets within each block of the media file simplifies the packet selection strategies of P2P systems. This fact allowed us to describe the receiver's buffer dynamics as an M/D/1 queue, and explicitly characterize the trade-off between the QoE metrics for different ranges of the system parameters.  We observed that the minimum initial buffer size to attain a desired level of interruption probability remains bounded as the file size grows if the arrival rate is slightly larger than the play rate. Further, when the arrival rate and the play rate match, the initial buffer size needs to scale as the square root of the file size to account for randomness of the arrivals. Finally, our numerical results confirmed that the optimal trade-off curves demonstrate a similar behavior to that predicted by our bounds.

This work is the first step in analytical characterization of QoE trade-offs in wireless media streaming applications.  An interesting extension to this work would be to obtain optimal resource allocation policies to satisfy users who have different interruption probability and initial waiting time targets.

\appendix

\noindent \textbf{Proof of Lemma \ref{r_bar_lemma}}:
\emph{Case I} $(0\leq R \leq 1)$: First note that $\gamma(r)$ is a continuously differentiable function, and $\gamma(0) = 0$. For each $R < 1$, we have $\gamma'(r) > 0$ for all $r \geq 0$. Therefore, $\gamma(r) > 0$ for all $r >0$, i.e., $\bar r(R) = 0$ for each $R<1$.

\emph{Case II} $(1\leq R \leq 2)$: By definition of $\bar r(R)$ in (\ref{r_bar}),
\begin{eqnarray*}
0 = \gamma(\bar r (R)) &=& \bar r(R) + R(e^{-\bar r(R)}-1) \\
&\leq&   \bar r(R) + R(-\bar r(R) + \frac{\bar r^2(R)}{2}).
\end{eqnarray*}
Rearranging the terms in the above relation, gives the lower bound in (\ref{case2}). We show the upper bound in two steps. First, we show that $\gamma(2(R-1)) >0$ for $R > 1$, then we verify that $\gamma(r) \geq 0$ for all $r\geq 2(R-1)$. These two facts imply that $\gamma(r) > 0$ for all $r \geq 2(R-1)$, i.e., $\bar r(R) \leq 2(R-1)$. The first step can be verified by noting that
$$\gamma(2(R-1)) \big|_{R=1} =0, \quad \frac{\partial}{\partial R}\gamma(2(R-1)) >0.$$
It is also straightforward to show that
\begin{equation}\label{gamma_der}
    \frac{\partial}{\partial r}\gamma(r) >0, \quad \foral r\geq \log(R),
\end{equation}
which immediately yields the second step by noting $r \geq 2(R-1) \geq \log(R)$.

\emph{Case III} $( R \geq 2)$: We use a similar technique as in the preceding case.  The upper bound is immediate by the following facts:
$$\gamma(R) = Re^{-R} >0, \quad  \frac{\partial}{\partial r}\gamma(r) >0,\quad  \foral r \geq R.$$
It is also straightforward to check that $\gamma(R-1) <0$ for all $R\geq 2$. Moreover, note that $\gamma(R) >0$. Therefore, by the intermediate value theorem, $\gamma(r)$ has a root in $[R-1, R]$, i.e., $\bar r(R) \geq R-1$.

\begin{lemma}\label{z_lemma}
For all $0 \leq z <1$, the following relation holds:
\begin{equation}\label{z_eq}
    -(1-z)\log(1-z) - z \leq -\frac{z^2}{2}.
\end{equation}
\end{lemma}

\begin{proof}
Let $f(z) = -(1-z)\log(1-z) - z + \frac{z^2}{2}$. $f(z)$ is a continuously differentiable function on $[0,1)$. Moreover, $f(0) = 0$, and $f'(z) =\log(1-z) +  z \leq 0$. Therefore, $f(z) \leq f(0) = 0$, for all $z \in [0,1)$.
\end{proof}

\begin{lemma}\label{Poisson_tail}
\emph{[Glynn] }Let $Z$ be a Poisson random variable with mean $\lambda$. If $\lambda \geq 2$, and $k \geq 2$, then
\begin{equation}\label{lower_tail}
    \pr \{Z \leq \lambda  - k \} \leq \exp\Big(\frac{1}{2\lambda}{\big(k - \frac32\big)^2}\Big).
\end{equation}
\end{lemma}
\begin{proof}
Glynn \cite{Glynn87} proves that for $\lambda \geq 2$ and $k \geq 2$
\begin{equation}\label{lower_tail_Q}
    \pr \{Z \leq \lambda  - k \} \leq b_\lambda Q\Big(\frac{k - \frac32}{\sqrt{\lambda}}\Big),
\end{equation}
where
$$Q(x) = \int_x^\infty \frac{1}{\sqrt{2\pi}} e^{-\frac{t^2}{2}} dt, \quad b_\lambda = (1+\frac1\lambda)\ e^{\frac{1}{8\lambda}}.$$

The result follows from the fact that $b_\lambda \leq 2$ for $\lambda \geq 2$, and
$$Q(x) \leq \frac{1}{2} \exp\big(-\frac{x^2}{2}\big), \quad \foral x \geq 0.$$
\end{proof}

\begin{lemma}\label{Poisson_tail_lb}
Let $Z$ be a Poisson random variable with mean $\lambda$. For all $m \leq \frac{\sqrt{\lambda}}{20}-1$
\begin{equation}\label{lower_tail_lb}
    \pr \{Z \leq \lambda  - m\sqrt{\lambda} \} \geq \frac13 \exp\Big(-\frac{1}{1.9}{\big(m +\frac12\big)^2}\Big).
\end{equation}
\end{lemma}

\begin{proof}
It follows from Proposition 6 of \cite{FoxGlynn88} that for all $0\leq i \leq \lambda/4$
\begin{equation}\label{poiss_lb}
    \pr \{Z = \lambda  - i \} \geq c_{\lambda} \exp\Big(-\frac{i^2}{1.9\lambda}\Big),
\end{equation}
where $c_{\lambda} =  \frac{e^{-\frac{1}{12}}}{\sqrt{2\pi\lambda}}$. Therefore, for $m \leq \sqrt{\lambda}/20-1$, we may verify that
\begin{eqnarray}
      \pr \{Z \leq \lambda  -  m\sqrt{\lambda}\} &\geq& c_{\lambda} \sum_{j = 0}^{\sqrt{\lambda}-1} \exp\Big(-\frac{(m\sqrt{\lambda} +j)^2}{1.9\lambda}\Big) \nonumber \\
      &\geq& c_{\lambda} e^{-\frac{m^2}{1.9}} \sum_{j = 0}^{\sqrt{\lambda}-1}  \exp\Big({-\frac{3m}{1.9\sqrt{\lambda}}j}\Big) \nonumber \\
      &=& c_{\lambda} e^{-\frac{m^2}{1.9}}  \frac{1-e^{-\frac{3m}{1.9}}}{1-e^{-\frac{3m}{1.9\sqrt{\lambda}}}} \nonumber \\
      &\geq& \Big( \frac{1.9e^{-\frac{1}{12}}}{3m\sqrt{2\pi}}(1-e^{-\frac{3m}{1.9}})\Big) e^{-\frac{m^2}{1.9}} \nonumber \\
      &\geq& \Big(\frac13 e^{-\frac{m}{1.9}-\frac{1}{7.6}}\Big)e^{-\frac{m^2}{1.9}} \nonumber \\
      &=& \frac13 \exp\Big(-\frac{1}{1.9}{\big(m +\frac12\big)^2}\Big). \nonumber
\end{eqnarray}
\end{proof}

\begin{lemma}\label{div_lemma}
Let $Q(t)$ be given by (\ref{buffer}). For $s, \delta > 0$, define the  boundary function
$$u(t) = D + \delta + (R+s-1)t.$$

If $s\leq R$, then the probability of crossing the boundary in any interval $[0,T]$ is bounded from above as 
$$\pr\Big\{\sup_{0\leq t\leq T}Q(t) \geq u(t) \Big\} \leq \exp\big(-\frac{s\cdot \delta}{R}\big). $$
\end{lemma}

\begin{proof}
Define $Z(t) = e^{rQ'(t)}$, where
\begin{equation}\label{Q'}
    Q'(t) = Q(t) - u(t) = -\delta + A(t) + (R+s)t,
\end{equation}
and $r > 0$ satisfies
$$\varphi(r) := e^r -1 - r(1+\frac{s}{R}) = 0.$$
Similarly to the proof of Lemma \ref{subMG_lemma}, we can show that $Z(t)$ is a martingale. This allows us to use Doob's maximal inequality to obtain
\begin{eqnarray}
  \pr\Big\{\sup_{0\leq t\leq T}Q(t) \geq u(t) \Big\} &=&  \pr\Big\{\sup_{0\leq t\leq T}Z(t) \geq 1 \Big\} \nonumber \\
  &\leq& \E[Z(T)] = e^{-r\cdot \delta}. \nonumber
\end{eqnarray}

Now it is sufficient to show that $r \geq \frac{s}{R}$. Observe that for all $0 \leq x \leq 1$
$$e^x \leq 1 + x + x^2.$$
Hence, for all  $s \leq R$ 
\begin{eqnarray}
\varphi(\frac{s}{R}) &=&  e^{\frac{s}{R}} -1 - \frac{s}{R}(1+\frac{s}{R}) \nonumber \\
&\leq& \frac{s}{R} + (\frac{s}{R})^2  - \frac{s}{R}(1+\frac{s}{R}) = 0.
\end{eqnarray}

Moreover, $\varphi(r) \rightarrow \infty$ when $r \rightarrow \infty$. Therefore, by intermediate value theorem there exists $r \geq \frac{s}{R}$ such that $\varphi(r) =0$. This completes the proof.
\end{proof}

\bibliographystyle{unsrt}
\bibliography{ISIT}

\begin{thebibliography}{10}

\bibitem{FraMoo03}
C.~Fraleigh, S.~Moon, B.~Lyle, C.~Cotton, M.~Khan, D.~Moll, R.~Rockell,
  T.~Seely, and C.~Diot.
\newblock Packet-level traffic measurements from the {S}print {IP} backbone.
\newblock {\em {IEEE} Network Magazine}, 17(6):6--16, 2003.

\bibitem{GumDun03}
K.~P. Gummadi, R.~J. Dunn, S.~Saroiu, S.~D. Gribble, H.~M. Levy, and
  J.~Zahorjan.
\newblock Measurement, modeling, and analysis of a peer-to-peer file-sharing
  workload.
\newblock In {\em Proc.\ {SOSP}}.

\bibitem{Lab09}
C.~Labovitz, D.~McPherson, and S.~Iekel-Johnson.
\newblock 2009 {I}nternet {O}bservatory report.
\newblock In {\em NANOG-47}, October 2009.

\bibitem{pplive}
{PPL}ive.
\newblock http://www.pplive.com/, 2009.

\bibitem{qqlive}
{QQL}ive.
\newblock http://www.qqlive.com/, 2009.

\bibitem{ZhoChi07}
Y.~P. Zhou, D.~M. Chiu, and J.~C.~S. Lui.
\newblock A simple model for analyzing {P2P} streaming protocols.
\newblock In {\em Proc. IEEE ICNP 2007}.

\bibitem{BonMas08}
T.~Bonald, L.~Massouli\'{e}, F.~Mathieu, D.~Perino, and A.~Twigg.
\newblock Epidemic live streaming: optimal performance trade-offs.
\newblock {\em SIGMETRICS Perform. Eval. Rev.}, 36(1):325--336, 2008.

\bibitem{ZhaLuiChi_09}
Bridge~Q. Zhao, John~C.S. Lui, and Dah-Ming Chiu.
\newblock Exploring the optimal chunk selection policy for data-driven {P2P}
  streaming systems.
\newblock In {\em The 9th International Conference on Peer-to-Peer Computing},
  2009.

\bibitem{YinSri10}
L.~Ying, R.~Srikant, and S.~Shakkottai.
\newblock {The Asymptotic Behavior of Minimum Buffer Size Requirements in Large
  P2P Streaming Networks}.
\newblock In {\em Proc. of the Information Theory and Applications Workshop (to
  appear)}, San Diego, CA, February 2010.

\bibitem{RLNC}
T.~Ho, R.~Koetter, M.~M\'edard, M.~Effros, J.~Shi, and D.~Karger.
\newblock A random linear network coding approach to multicast.
\newblock {\em IEEE Transactions on Information Theory}, 52:4413--4430, 2006.

\bibitem{Acedanski05}
S.~Acedanski, S.~Deb, M.~M\'edard, and R.~Koetter.
\newblock How good is random linear coding based distributed networked storage.
\newblock In {\em {NetCod}}, 2005.

\bibitem{Rodriguez06}
C.~Gkantsidis, J.~Miller, and P.~Rodriguez.
\newblock Comprehensive view of a live network coding p2p system.
\newblock In {\em Proc. {ACM SIGCOMM}}, 2006.

\bibitem{wangLi07}
M.~Wang and B.~Li.
\newblock R$^2$: Random push with random network coding in live peer-to-peer
  streaming.
\newblock {\em IEEE JSAC, Special Issue on Advances in Peer-to-Peer Streaming
  Systems}, 25:1655--1666, 2007.

\bibitem{ChiZhang06}
H.~Chi and Q.~Zhang.
\newblock Deadline-aware network coding for video on demand service over p2p
  networks.
\newblock In {\em PacketVideo}, 2006.

\bibitem{SP_book}
I.~Karatzas and S.~Shreve.
\newblock {\em Brownian Motion and Stochastic Calculus}.
\newblock Springer, 1997.

\bibitem{ISIT_report}
A.~ParandehGheibi, M.~M\'edard, S.~Shakkottai, and A.~Ozdaglar.
\newblock Avoiding interruptions - {QoE} trade-offs in block-coded streaming
  media applications.
\newblock arXiv:1001.1937 [cs.MM].

\bibitem{Glynn87}
P.~W. Glynn.
\newblock Upper bounds on poisson tail probabilities.
\newblock {\em Operations Research Letters}, 6(1):9--14, 1987.

\bibitem{FoxGlynn88}
B.~L. Fox and P.~W. Glynn.
\newblock Computing poisson probabilities.
\newblock {\em Communications of the {ACM}}, 31(4):440--445, 1988.

\end{thebibliography}

%

\end{document}